\documentclass[reqno]{amsart}
\usepackage[english]{babel}
\usepackage{amssymb}
\usepackage[T1]{fontenc}
\usepackage{graphicx,float}
\usepackage{subcaption}
\usepackage{mwe}
\usepackage{xcolor,hyperref,verbatim}
\newtheorem{theorem}{Theorem}[section]

\newtheorem{lemma}[theorem]{Lemma}

\theoremstyle{definition}

\newtheorem{example}[theorem]{Example}
\theoremstyle{remark}
\newtheorem*{remark}{Remark}
\numberwithin{equation}{section}

\newcommand{\uu}{\pmb{u}}
\newcommand{\ff}{\pmb{f}}

\newcommand{\nn}{\pmb{\hat{n}}}
\newcommand{\Id}{\pmb{1}}

\newcommand{\CC}{\mathbb{C}}
\newcommand{\DD}{\mathbb{D}}
\newcommand{\pdzb}{\partial_{\overline{z}}}

\newcommand{\GG}{\pmb{g}}

\newcommand{\zz}{\pmb{\hat{z}}}
\newcommand{\zb}{\overline{z}}\newcommand{\wb}{\overline{w}}

\newcommand{\A}{\mathbb{A}_r}

\newcommand{\MM}{\pmb{M}}
\begin{document}
\title{Explicit Solutions in Isotropic Planar Elastostatics}

\author{Andreas Granath}
\address{Department of Mathematics and Mathematical Statistics\\ Ume\aa \ University\\ SE-901 87 Ume\aa \\ Sweden}
\email{andreas.granath@umu.se}
\author{Per \AA{hag}}
\address{Department of Mathematics and Mathematical Statistics\\ Ume\aa \ University\\ SE-901 87 Ume\aa \\ Sweden}
\email{per.ahag@umu.se}
\author{Antti Per\"{a}l\"{a}}
\address{Department of Mathematics and Mathematical Statistics\\ Ume\aa \ University\\ SE-901 87 Ume\aa \\ Sweden}
\email{antti.perala@umu.se}

\author{Rafa\l\ Czy{\.z}}
\address{Institute of Mathematics \\ Faculty of Mathematics and Computer Science \\ Jagiellonian University\\ \L ojasiewicza 6\\ 30-348 Krak\'ow\\ Poland}
\email{Rafal.Czyz@im.uj.edu.pl}


\begin{abstract}
Addressing the intricate challenges in plane elasticity, especially with non-vanishing traction and complex geometries, requires innovative methods. This paper offers a novel approach, drawing inspiration from the Neumann problem for the inhomogeneous Cauchy-Riemann equations. Our method applies to domains conformally equivalent to a unit disk or an annulus, focusing on deriving explicit solutions for the displacement field rather than the stress tensor, which distinguishes it from most traditional approaches. We explore solutions for specific classical cases to demonstrate its efficacy, such as a cardioid domain, a ring domain with a shifted hole, and a gear-like structure. This work enhances the toolkit for researchers and practitioners tackling isotropic planar elastostatic challenges with a unified and flexible approach.
\end{abstract}

\maketitle

\section{Introduction}

The study of plane elasticity in a domain $\Omega$ has a vibrant history in the engineering and physical sciences. Here, $\Omega$ represents a material that is symmetric, isotropic, and linear in its elastic behavior,  traditionally captured by two equations:
\begin{align}
-\nabla\cdot\sigma(\uu)&=\ff\text{ in }\Omega,\label{eq:elasticityproblem} \\
\sigma(\uu)\cdot\nn&=\GG\text{ on }\partial\Omega.
\nonumber
\end{align}
In these equations, $\sigma(\uu)$ is the stress tensor, $\ff$ denotes an applied force, and $\nn$ is the outward pointing unit normal. The stress tensor $\sigma(\uu)$ is then described by $\sigma(\uu)=2\mu\varepsilon(\uu)+\lambda(\nabla\cdot\uu)\Id$ with $\Id$ being the $2\times 2$ identity matrix, and $\varepsilon(\uu)$ is the strain tensor given as the symmetric part of the Jacobian of $\uu$. Furthermore, $\mu, \lambda$ are the associated Lamé parameters, and the force $\ff$ can be expressed as $\ff=\rho\pmb{F}$, where $\pmb{F}$ is a specific force and $\rho$ is the material density. Finally, $\GG$ represents the given boundary traction on $\partial\Omega$.

These equations are typically addressed using numerical solutions like the finite element or boundary element methods. However, complexities in geometric properties often pose substantial difficulties. Hence, deriving explicit solutions has been an area of considerable focus.

The use of complex functions in plane elastic problems was initiated by Kolosov in 1909 \cite{Kolosov1}, building on Kirsch’s method, which addressed solutions around a circular hole \cite{Kirsch}. Several authors have extended this work, most famously Muskhelishvili \cite{Mushk} for more general domains and Westergaard~\cite{Westergaard} for crack problems. Both of these methods build on representing the stress tensor in terms of Airy potential functions, allowing the stresses to be solved. It is then possible, in some instances, to solve for the displacement via these stresses. It was also shown by Yun \cite{Yun} in the 1990s that it is possible to obtain potential-based solutions to the elasticity problem starting with the Navier equation.

However, these approaches do come with their drawbacks. In this work, we demonstrate a novel method for solving a broad class of problems in plane elasticity. In particular, we establish a  correspondence between the well-studied Neumann problem for the inhomogeneous Cauchy-Riemann equations and the problem of linear elasticity in bounded plane domains that are conformally equivalent to either the disk or an annulus. Although reminiscent of other complex methods, such as the ones in \cite{Constanda,Kolosov1,Mushk,Neuber}, our method yields a more direct approach as well as an explicit formula for the displacement field rather than the stress tensor.

An additional area for improvement that has emerged with this method pertains to the challenges that arise when considering ring-shaped domains. These structures are common in many engineering constructions, e.g., gears and ring bearings. To study deformation and internal stress distribution in these, one often uses bipolar coordinates, e.g., as in Chen et al. \cite{Chen1,Chen2} or Alaci, Ciornei, and Romanu \cite{Alaci}. We shall also consider a simple notch problem that has historically been challenging when using Airy potentials  due to  singularities at the corner \cite{Helsing,Williams}.

Motivated by these challenges, this paper outlines a new method for obtaining integral representations of the solution to \eqref{eq:elasticityproblem} in domains that are images of conformal maps from either a unit disk or an annulus. We accomplish this by leveraging the established knowledge of the Neumann problem for the inhomogeneous Cauchy-Riemann equations and the known representations of its solutions due to Begehr and Vaitekhovich \cite{Begehr,Vaitekhovich}.

Our main results, encapsulated in Theorem~\ref{theorem:SimplyConnectedMain} for simply connected domains $\Omega$ that are conformally equivalent to the unit disk $\DD$ and Theorem~\ref{theorem: SolutionAnnularDomain} for domains $\Omega$ that are conformally equivalent to the annulus
$\A=\{ z\in\CC : r<|z|< 1 \}$, provide advancements in solving plane elasticity problems. In Section~\ref{sec:Appl}, we further illustrate the applicability of our method by demonstrating solutions to a notch problem in a cardioid domain, a ring domain with a shifted hole, and a gear-like domain, underscoring the wide-ranging potential of our approach.

When analyzing the method presented in this paper, it is crucial to contextualize it within the broader landscape of existing techniques. A natural point of comparison is the classical method, commonly referred to as Kolosov's method or the classical complex method \cite{Mushk}. This method boasts the ability to explicitly determine stress values across diverse domains. A disadvantage of Kolosov's method is that it requires the conformal map to be rational. Although attempts have been made to dispense with this necessity by seeking approximations through rational functions, these endeavors have proven cumbersome and not particularly fruitful, as Neuber \cite{Neuber} pointed out. Moreover, the classical method often necessitates deriving the primary ingredients for individual problems since a comprehensive framework for general force classes or domains is absent. This contrasts with our method, which sets conditions on general classes of both loads and domains.  Finally, when compared to the real integral equation method \cite{Constanda} and the Somigliana method \cite{Rizzo}, a distinct advantage of our approach becomes evident as it sidesteps the need to introduce intermediary entities like integral operators or layer potentials.

\section{Notation and preliminaries}

In this section, we introduce several essential concepts and notation used throughout the paper. Our focus is on the elastostatic equation in a bounded planar domain $\Omega$ for a specific material with associated Lamé parameters $\mu,\lambda$, which describe the elastic properties of the material. The problem is typically presented in its divergence form~\eqref{eq:elasticityproblem}. The displacement field $\uu$ describes the deformation of the domain, $\ff$ denotes a body force acting on the domain, and $\GG$ represents the boundary traction on $\partial\Omega$. The force $\ff$ can also be expressed in terms of a specific force $\pmb{F}$ as $\ff=\rho\pmb{F}$. The representation of the force will turn out to be important, hence we shall consider the standard representation $\pmb{F}=\nabla\phi+\nabla\times\pmb{A}$.

If $\Omega$ has multiple boundaries, $\GG$ is a piecewise-defined function. The stress tensor $\sigma(\uu)$ is a matrix derived from the displacement field $\uu$, defined as $\sigma(\uu)=2\mu\varepsilon(\uu)+\lambda(\nabla\cdot\uu)\Id$ with $\Id$ being the $2\times 2$  identity matrix, and $\varepsilon(\uu)$ representing the strain tensor. The strain tensor is the symmetric part of the Jacobian of $\uu$, that is $\varepsilon(\uu)=\frac{\nabla\uu+(\nabla\uu)^T}{2}$. For an in-depth introduction to this problem and its generalizations, we refer the reader to Chapter 7 in \cite{Lurie}.

The divergence form, although concise, can be challenging to work with. Thus, it is common to rewrite~\eqref{eq:elasticityproblem}  into an explicit second-order form, known as the Navier equation or the Navier-Cauchy equation:
\begin{equation}
(\mu+\lambda)\nabla(\nabla\cdot\uu)+\mu\nabla^2\uu+\rho\pmb{F}=0.
\label{eq:Navier}
\end{equation}
Here, $\nabla^2$ denotes the Laplacian, and the body force $\ff$ has been expressed in terms of a specific force $\pmb{F}$ using the material density $\rho$. This conversion is carried out in the proof of Theorem~\ref{theorem:EquivalenceSimplyConnected} and will be needed to link the above problem with the Neumann problem for the inhomogeneous Cauchy-Riemann equations.

It is worth noting the relationship between \eqref{eq:Navier} and the dynamical problem. The equation~\eqref{eq:Navier} can be regarded as the steady-state equation for elastodynamics. Adding the term $\rho\frac{\partial^2\uu}{\partial t^2}$ to the right-hand side yields the standard equation describing the dynamical behavior of elastic solids. Therefore, if the equation is separable, the results in this paper may facilitate the construction of dynamical solutions.

Next, we recall the notation for the classical Wirtinger derivatives:
$$\pdzb = (\partial_x + i\partial_y)/2 \quad \text{ and } \quad \partial_z = (\partial_x - i\partial_y)/2.$$
For a planar set $E$, we denote by $\mathcal{C}(E;\CC)$ the space of continuous complex-valued functions defined on $E$. For a real number $\alpha$ such that $0 < \alpha < 1$, the notation $\mathcal{C}^\alpha(E;\CC)$ denotes the space of $\alpha$-Hölder continuous functions. These are functions $\psi$ that satisfy
$$|\psi(z)-\psi(\zeta)|\leq C|z-\zeta|^\alpha,$$
for some $C>0$. For positive integer values of $k$, $\mathcal{C}^k(E;\CC)$ denotes functions whose partial derivatives up to order $k$ are continuous in $E$, while $\mathcal{C}^k(\overline{E},\CC)$ denotes those that are $\mathcal{C}^k$ in a neighborhood of $E$. Lastly, the notation $\mathcal{C}^{1,\alpha}(E;\CC)$ denotes functions whose first partial derivatives belong to the space $\mathcal{C}^{\alpha}(E;\CC)$.

\section{The case of simply connected domains}

In this section, we demonstrate that the plane elasticity problem, as shown in \eqref{eq:elasticityproblem}, in a simply connected domain can be understood as an equivalent Neumann problem for the inhomogeneous Cauchy-Riemann equations within the same domain, viewed in $\CC$. First, we need to recall the following result.

\begin{theorem}\label{thm:Begehr}
Let $\psi\in\mathcal{C}^\alpha(\overline{\DD};\CC)$, $0<\alpha<1$, $\gamma\in \mathcal{C}(\partial\DD;\CC)$, $c\in\CC$. The Neumann problem for the inhomogeneous Cauchy-Riemann equations in the unit disk:
\begin{align*}
    \pdzb w&=\psi\text{ in }\DD,\\
    \partial_\nu w&=\gamma\text{ on }\partial\DD\\
    w(0)&=c
\end{align*}
is solvable if and only if for $|z|<1$
$$\frac{1}{2\pi i}\int_{|\zeta|=1}\gamma(\zeta)\frac{d\zeta}{(1-\overline{z}\zeta)\zeta}+\frac{1}{2\pi i}\int_{|\zeta|=1}\psi(\zeta)\frac{d\overline{\zeta}}{1-\overline{z}\zeta}+\frac{1}{\pi}\int_{|\zeta|<1}\frac{\overline{z}\psi(\zeta)}{(1-\overline{z}\zeta)^2}d\xi d\eta=0.$$
The unique solution is then given by
$$w(z)=c-\frac{1}{2\pi i}\int_{|\zeta|=1}\big(\gamma(\zeta)-\overline{\zeta}\psi(\zeta)\big)\log(1-z\overline{\zeta})\frac{d\zeta}{\zeta}-\frac{1}{\pi}\int_{|\zeta|<1}\frac{z\psi(\zeta)}{\zeta(\zeta-z)}d\xi d\eta.$$
\end{theorem}
\begin{proof}
See Theorem 11 in \cite{Begehr}.
\end{proof}

 Next, we turn to the first contribution of this paper.

\begin{theorem}\label{theorem:EquivalenceSimplyConnected}
    Let $\mathbb{D}$ denote the unit disk and $\mu,\lambda$ its associated Lamé parameters. Let $\phi$, $\pmb{A}=A\zz$, where $\phi,A$, are $\mathcal{C}^2$ in a neighborhood of $\mathbb{D}$, be potentials and $\pmb{F}=~\nabla\phi+\nabla\times\pmb{A}$ denote the associated specific loading force. Finally, let the boundary traction $\GG$ be $\mathcal{C}^1$ in a neighborhood of $\partial\mathbb{D}$. Then, the existence of the displacement field $\uu = [u_1,u_2]^T$ that solves
\begin{align}
    -&\nabla\cdot\sigma(\uu)=\rho\pmb{F}\text{ in }\mathbb{D}\nonumber\\
    &\sigma(\uu)\cdot\nn=\GG\text{ on }\partial \mathbb{D}\label{thm_D_Ph}\\
    & \uu(0)=\pmb{0}\nonumber
\end{align}
    is equivalent to the existence of a solution $v = u_1 - i u_2$ to the Neumann problem for the inhomogeneous Cauchy–Riemann equations
    \begin{align}
        \pdzb v&=\frac{\psi}{2}\text{ in }\mathbb{D}\nonumber\\
        \partial_{\nu} v&=\gamma\text{ on }\partial\mathbb{D}\label{thm_D_math}\\ \nonumber
        v(0)&=0,
    \end{align}
    where $\psi=-\frac{\rho}{2\mu+\lambda}\phi-i\frac{\rho}{\mu}A$, $\gamma=\frac{1}{2\mu}(g_1-ig_2)+\frac{\rho}{2\mu}\overline{\zeta}\big(\nu\phi-iA\big)$, and $\nu=\frac{\lambda}{\lambda+2\mu}$ is the Poisson ratio.

\end{theorem}

\begin{proof} We begin by assuming that  \eqref{thm_D_Ph} holds. Let $(x_0,y_0)\in\DD$ and $\DD_\epsilon=\DD_\epsilon(x_0,y_0)$ be a disk of radius $\epsilon$ such that $\overline{\DD}_\epsilon\subset\DD$. Assuming planarity, i.e., $\uu=\uu(x,y)$ and $\uu\cdot\zz=0$, together with the standard Laplacian identity $\nabla^2\uu=\nabla(\nabla\cdot\uu)-\nabla\times(\nabla\times\uu)$ allows us to rewrite and integrate \eqref{eq:Navier} to obtain the equality
\begin{align}
    \int_{\DD_\epsilon}[(2\mu+\lambda)\nabla(\nabla\cdot\uu)-\mu\nabla\times(\nabla\times\uu)+\rho\pmb{F}]\, dA=\pmb{0}\qquad \text{ in }\DD_\epsilon.
    \label{eq:associationform}
\end{align}
Next, we associate the displacement field $\uu$ with the potentials in $\pmb{F}$. We recall the following integral identities,
\begin{equation*}
    \int_\DD\nabla f dA=\int_{\partial\DD}f\nn \, ds\quad\text{ and }\quad \int_\DD\nabla\times(g\zz)dA=-\int_{\partial\DD} g\zz\times\nn \, ds ,
\end{equation*}
which allows us to rewrite \eqref{eq:associationform} as a line integral along the closed boundary of $\DD_\epsilon$ using that $\pmb{F}=\nabla\phi+\nabla\times(A\zz)$, resulting in
\begin{equation}
    \int_{\partial\DD_\epsilon} \left((2\mu+\lambda)(\nabla\cdot\uu)\nn+\mu(\nabla\times\uu)\times\nn+\rho\phi\nn-\rho A\zz\times\nn\right)\, ds=\pmb{0}.
\end{equation}
We note that the integrand can be expressed as
\begin{equation*}
    \int_{\partial\DD_\epsilon}\begin{bmatrix}
        (2\mu+\lambda)(\nabla\cdot\uu)+\rho\phi & -\mu(\nabla\times\uu)\cdot\zz+\rho A \\
        \mu(\nabla\times\uu)\cdot\zz-\rho A & (2\mu+\lambda)(\nabla\cdot\uu)+\rho\phi
    \end{bmatrix}\begin{bmatrix}
        n_1 \\ n_2
    \end{bmatrix}ds=\pmb{0}.
\end{equation*}
This now implies two possibilities: either the terms in the integrand vanish, and we have the desired associations, or the matrix must correspond to a divergence-free vector field. Assume the latter case, let $h_1=(2\mu+\lambda)\nabla\cdot\uu+\rho\phi$ and $h_2=\mu(\nabla\times\uu)\cdot\zz-\rho A$. Then the divergence theorem gives
\begin{align}
    \pmb{0} = \int_{\partial \DD_\epsilon} \begin{bmatrix}
        h_1 & -h_2 \\ h_2 & h_1
    \end{bmatrix} \begin{bmatrix}
        n_1 \\ n_2
    \end{bmatrix}
    ds = & \int_{\DD_\epsilon} \nabla \cdot \begin{bmatrix}
        h_1 & -h_2 \\ h_2 & h_1
    \end{bmatrix} ds \nonumber \\
    &= \int_{\DD_\epsilon} \begin{bmatrix}
        \partial_x h_1 - \partial_y h_2 \\
        \partial_x h_2 + \partial_y h_1
    \end{bmatrix} ds.
    \label{eq:diveq}
\end{align}
Now, observe that~\eqref{eq:diveq} holds for any piecewise $\mathcal{C}^1$-smooth subdomain in $\DD_\epsilon$. Therefore, if we let $p=h_1+ih_2$ and $R\subset\DD_\epsilon$ an arbitrary rectangle we get
\begin{align*}
    \int_{\partial R}p(z)dz&=\int_{\partial R}(h_1dx-h_2dy)+i\int_{\partial R}(h_2dx+h_1dy)\\
    &=-\int_R(\partial_xh_2+\partial_yh_1)dxdy+i\int_R(\partial_xh_1-\partial_yh_2)dxdy=0,
\end{align*}
which by Morera's theorem then implies that $p(z)$ is analytic in $\DD_\epsilon$, hence $h_1$ and $h_2$ are harmonic in $\DD_\epsilon$.

Since the normal to $\partial\DD_\epsilon$ is given by $\nn=\frac{1}{\epsilon}[x,y]^T$, we then have, by the mean-value property of harmonic functions, the following equality from \eqref{eq:diveq}
$$
\begin{aligned}
&h_1(x_0,y_0)x_0-h_2(x_0,y_0)y_0=\frac{1}{2\pi \epsilon}\int_{\partial\DD_\epsilon} (h_1x-h_2y)ds=0, \\
&h_2(x_0,y_0)x_0+h_1(x_0,y_0)y_0=\frac{1}{2\pi \epsilon}\int_{\partial\DD_\epsilon}(h_2x+h_1y)ds=0 .
\end{aligned}
$$
Therefore, $h_1(x_0,y_0)=h_2(x_0,y_0)=0$. Since $(x_0,y_0)\in \DD$ was arbitrary, this yields that $h_1=h_2\equiv 0$ throughout $\DD$.

Using the result above for $h_1$ and $h_2$ we can infer a relation between the displacement $\uu$ and the potentials $\phi,A$. From this, we proceed to derive the complex problem. Upon using  that $h_1$ and $h_2$ vanish and their respective definitions we can make the substitutions $u_1=v_1, u_2=-v_2$ to get
\begin{align*}
    &\partial_xv_1-\partial_yv_2=-\frac{\rho}{2\mu+\lambda}\phi,\\
    &\partial_xv_2+\partial_yv_1=-\frac{\rho}{\mu}A.
\end{align*}
Multiplying the second equality above with the imaginary unit and summing, we get
\begin{equation*}
    (\partial_xv_1-\partial_yv_2)+i(\partial_xv_2+\partial_yv_1)=2\pdzb v=-\frac{\rho}{2\mu+\lambda}\phi-i\frac{\rho}{\mu}A,
\end{equation*}
which can be rewritten as
\begin{equation*}
    \pdzb v=-\frac{\rho}{2}\bigg(\frac{1}{2\mu+\lambda}\phi+\frac{i}{\mu}A\bigg)=\frac{\psi_1+i\psi_2}{2}=\frac{\psi}{2}.
\end{equation*}
For the normal derivative, straightforward calculations give
\begin{multline*}
    \sigma(\uu)\cdot\nn=\mu(\nabla\uu+\nabla\uu^T)\cdot\nn+\lambda(\nabla\cdot\uu)\Id\cdot\nn \\ =
    \bigg[\begin{pmatrix}
        2\mu \partial_x u_{1} & \mu(\partial_yu_{1}+\partial_x u_{2})\\
        \mu(\partial_yu_{1}+\partial_xu_{2}) & 2\mu \partial_y u_{2}
    \end{pmatrix}
    +\begin{pmatrix}
    \lambda(\partial_x u_{1}+\partial_yu_{2}) & 0\\ 0 & \lambda(\partial_xu_{1}+\partial_yu_{2})
    \end{pmatrix}\bigg]\cdot\nn\\
    =\bigg[\begin{pmatrix}
    2\mu \partial_xu_{1} & 2\mu \partial_yu_{1}-\mu\psi_2 \\
    2\mu \partial_xu_{2}+\mu\psi_2 & 2\mu \partial_yu_{2}
    \end{pmatrix}+\begin{pmatrix}
    \lambda\psi_1 & 0\\
    0 & \lambda\psi_1
    \end{pmatrix}\bigg]\cdot\nn\\
    =2\mu\nabla\uu\cdot\nn+\begin{pmatrix}
        \lambda\psi_1 & -\mu\psi_2 \\ \mu\psi_2 & \lambda\psi_1
    \end{pmatrix}\cdot\nn=2\mu\nabla\pmb{v}^*\cdot\nn+\MM\cdot\nn =\GG ,
\end{multline*}
where $\pmb{v}^*=[v_1,-v_2]^T$ and $\MM$ is introduced such that $\MM_{11}=\MM_{22}=-\rho\nu\phi$ and $\MM_{12}=-\MM_{21}=\rho A$. From the boundary condition on $v$, we find as required
\begin{equation}
\gamma=\gamma_1+i\gamma_2=\frac{1}{2\mu}(g_1-ig_2)+\frac{\rho}{2\mu}\left(\nu\phi-iA\right)(n_1-in_2),
\label{eq:gamma}
\end{equation}
where $n_1,n_2$ denotes the components of $\nn$. In particular, as we consider the unit disk we have $n_1-in_2=\overline{\zeta}$ and the implication from \eqref{thm_D_Ph} to \eqref{thm_D_math} follows. We are now set to prove the reverse implication. Assume that the function $v$ satisfies \eqref{thm_D_math}, i.e.
    \begin{align*}
        \pdzb v&=\frac{\psi}{2}\text{ in }\mathbb{D}\nonumber\\
        \partial_{\nu} v&=\gamma\text{ on }\partial\mathbb{D}\label{thm_D_math}\\ \nonumber
        v(0)&=0,
    \end{align*}
where $\gamma$ and $\psi$ are defined above. To complete the proof, we shall verify that the field $\uu=[v_1,-v_2]^T$ satisfies Navier's equation~\eqref{eq:Navier}. Recalling that the Wirtinger derivative can be expanded as
$$\pdzb v=\frac{1}{2}(\partial_x+i\partial_y)(v_1+iv_2)=\frac{1}{2}(\partial_xv_1-\partial_yv_2)+\frac{i}{2}(\partial_xv_2+\partial_yv_1)=\frac{\psi_1+i\psi_2}{2}.$$
 and equating real and imaginary parts in the expression for $\pdzb v$, we get that $\nabla\cdot\uu=\psi_1$ and $-(\nabla\times\uu)=\psi_2\zz$. Substituting these equalities into Navier's equation and using the identity $\nabla^2\uu=\nabla(\nabla\cdot\uu)-\nabla\times(\nabla\times\uu)$, we get the left hand side
$$(2\mu+\lambda)\nabla\psi_1+\mu\nabla\times(\psi_2\zz)+\rho F.$$
Upon substituting $\psi_1=-\frac{\rho}{2\mu+\lambda}\phi$ and $\psi_2=-\frac{\rho}{\mu}A$ into the above equation and utilizing that $F=\nabla\phi+\nabla\times(A\zz)$ we note that the term vanishes, and hence the equation is satisfied. It now remains to show that the boundary data is correct. This can be seen from a similar calculation. The condition on the normal derivative can be reformulated as
$$\partial_\nu v=(x\partial_x+y\partial_y)(v_1+iv_2)=\nn\cdot\nabla(v_1+iv_2)=\gamma_1+i\gamma_2.$$
Hence the components of the vector field $\uu$ satisfies $\nabla u_1\cdot\nn=\gamma_1$ and $\nabla u_2\cdot\nn=-\gamma_2$. Using our previous relation between $\sigma(\uu)\cdot\nn$ and the gradient of $\uu$ we see that

 \begin{equation}
     \sigma(\uu)\cdot\nn=2\mu\nabla\uu\cdot\nn+\begin{pmatrix}
         \lambda\psi_1 && -\mu\psi_2\\ \mu\psi_2 && \lambda\psi_1
     \end{pmatrix}\cdot\nn=\begin{pmatrix}
         2\mu\gamma_1+\lambda\psi_1x-\mu\psi_2y \\
         -2\mu\gamma_2+\mu\psi_2x+\lambda\psi_1y
     \end{pmatrix},
     \label{eq:bdrydata}
 \end{equation}
 where we have used that the normal on $\partial\DD$ is given explicitly by $\nn=[x,y]^T$. Recalling the definition of $\gamma$ we see that the real and imaginary components can be written as
 \begin{align*}
     \gamma_1&=\frac{g_1}{2\mu}+\frac{\rho}{2\mu}(x\nu\phi-yA),\\
     \gamma_2&=-\frac{g_2}{2\mu}-\frac{\rho}{2\mu}(y\nu\phi+xA).
 \end{align*}
 Substituting the above expressions into \eqref{eq:bdrydata} and using that the Poisson ratio is given by $\nu=\frac{\lambda}{2\mu+\lambda}$ we obtain that $\sigma(\uu)\cdot\nn=[g_1,g_2]^T$ as desired, proving the implication from \eqref{thm_D_math} to \eqref{thm_D_Ph}.
\end{proof}
\begin{remark}
    The construction of the function $\gamma$ in \eqref{eq:gamma} is not inherent for the unit disk, but can be applied for general domains where the normal $\nn$ is known. We shall use this when considering general transformed domains.
    \label{remark:normal}
\end{remark}

We will employ biholomorphic mappings to extend Theorem~\ref{theorem:EquivalenceSimplyConnected} to the setting of simply connected domains in Theorem~\ref{theorem:SimplyConnectedMain}. A crucial step in achieving this generalization is provided by Lemma~\ref{lemma:transformedsimplyconnected}.

\begin{lemma}\label{lemma:transformedsimplyconnected}
Let $\Omega$ be a simply connected domain with $\partial\Omega$ being a smooth Jordan curve, $\varphi$ a biholomorphic mapping from a neighborhood of $\overline{\DD}$ to a neighborhood of $\overline{\Omega}$, such that $\varphi(\DD)=\Omega$ and $\varphi(\partial \DD)=\partial\Omega$, $\psi,\gamma\in\mathcal{C}^2(\overline{\Omega};\CC)$.
Moreover, assume that there is a function $v\in \mathcal{C}^2(\overline{\DD};\CC)$ which satisfies
\begin{align*}
    \pdzb v &= \frac{1}{2}\big(\psi\circ\varphi\big)\overline{\partial_z\varphi}\quad\text{ in }\DD,\\
    \partial_\nu v &= |\varphi_z|(\gamma\circ\varphi)\text{ on }\partial\DD.
\end{align*}
Then $u=v\circ\varphi^{-1}\in\mathcal{C}^2(\overline{\Omega};\CC)$ solves
\begin{align*}
    \partial_{\bar{w}}u &= \frac{\psi}{2}\text{ in }\Omega,\\
    \partial_\nu u &= \gamma\text{ on }\partial\Omega.
\end{align*}
\end{lemma}
\begin{proof}
We begin by defining $u=v\circ\varphi^{-1}$. From this definition, it follows that
$$\partial_{\wb}u=(v_{\zb}\circ\varphi^{-1})\overline{\varphi_w^{-1}(w)}=\frac{1}{2}\psi(w)(\overline{\varphi_z}\circ\varphi^{-1})\overline{\varphi_w^{-1}}=\frac{1}{2}\psi(w),$$
as desired. Next, let $\beta(w,\wb)=\varphi^{-1}(w)\overline{\varphi^{-1}(w)}-1$ be the function defining the boundary of the unit disk. Then, the outward unit normal on the boundary of $\Omega$ is given by
\begin{equation}\label{lem:SCD_OUN}
\partial_\nu u=\frac{\beta_{\wb}}{|\beta_{\wb}|}\partial_wu+\frac{\beta_w}{|\beta_w|}\partial_{\wb}u\quad \text{ on }\partial\Omega.
\end{equation}
By substituting $w=\varphi(z)$ in~\eqref{lem:SCD_OUN} we arrive at
$$\partial_\nu v=z\partial_zv+\zb\partial_{\zb}v=|\varphi_z|(\gamma\circ\varphi)\quad\text{ on }\partial\DD ,$$
and by employing the chain rule it then follows
\begin{equation}
    \begin{split}
        \partial_\nu u&=\frac{\beta_{\wb}}{|\beta_{\wb}|}(v_z\circ\varphi^{-1})\varphi_w^{-1}+\frac{\beta_w}{|\beta_w|}\frac{\psi}{2}\\
        &=\frac{\beta_{\wb}}{|\beta_{\wb}|}\bigg(\frac{|\varphi_z|\circ\varphi^{-1}}{\varphi^{-1}}\gamma-\frac{\overline{\varphi^{-1}}}{\varphi^{-1}}(\overline{\varphi_z\circ\varphi^{-1}})\frac{\psi}{2}\bigg)\varphi^{-1}_w+\frac{\beta_w}{|\beta_w|}\frac{\psi}{2}\\
        &=\frac{\beta_{\wb}}{|\beta_{\wb}|}\frac{\varphi^{-1}_w}{|\varphi^{-1}_w|\varphi^{-1}}\gamma-\frac{\beta_{\wb}}{|\beta_{\wb}|}\frac{\overline{\varphi^{-1}}\varphi^{-1}_w}{\overline{\varphi^{-1}_w}\varphi^{-1}}\frac{\psi}{2}+\frac{\beta_w}{|\beta_w|}\frac{\psi}{2}=\gamma(w).
   \end{split}
\end{equation}
Here, it was used that $\beta_w=\varphi^{-1}_w\overline{\varphi^{-1}}$, $\beta_{\wb}=\overline{\varphi^{-1}_w}\varphi^{-1}$ and $|\beta_w|=|\beta_{\wb}|=|\varphi^{-1}_w|$. Thus, the proof is completed.
\end{proof}

 With the help of Theorem~\ref{theorem:EquivalenceSimplyConnected} and Lemma~\ref{lemma:transformedsimplyconnected}, we shall now finalize the proof of the main result of this section.

\begin{theorem}\label{theorem:SimplyConnectedMain}
Let $\Omega$ be a simply connected domain with $\partial\Omega$ being a smooth Jordan curve, $\varphi$ a biholomorphic mapping from a neighborhood of $\overline{\DD}$ to a neighborhood of $\overline{\Omega}$, such that $\varphi(\DD)=\Omega$, $\varphi(\partial \DD)=\partial\Omega$, $\varphi(0)=0$, and with associated Lamé parameters $\mu,\lambda$. Furthermore, let $\phi$, $\pmb{A}=A\zz$ be potentials, where $\phi,A$, are $\mathcal{C}^2$ in a neighborhood of $\Omega$ and $\pmb{F}=~\nabla\phi+\nabla\times\pmb{A}$ denote the associated specific loading force. Finally, let the boundary traction $\GG=[g_1, g_2]^T$ be $\mathcal{C}^1$ in a neighborhood of $\partial\Omega$. Then, there exists a unique solution to the plane elasticity problem
\begin{align*}
    -\nabla\cdot\sigma(\uu)&=\ff\text{ in }\Omega\\
\sigma(\uu)\cdot\nn&=\GG\text{ on }\partial\Omega\\
    \uu(0)&=\pmb{0}.
\end{align*}
Moreover, the solution is given by $\uu = [u_1, u_2]^T$, where $u_1$ and $u_2$ denote the positive real and negative imaginary parts of $v\circ\varphi^{-1}$, respectively, and $v$ is given by

\begin{align*}
    v(z)=-\frac{1}{2\pi i}\int_{|\zeta|=1}\bigg(|\varphi_\zeta(\zeta)|(\gamma\circ\varphi(\zeta)&-\frac{\overline{\zeta}}{2}(\psi\circ\varphi(\zeta))\overline{\varphi_\zeta(\zeta)}\bigg)\log(1-z\overline{\zeta})\frac{d\zeta}{\zeta}\\
    &-\frac{z}{\pi}\int_{|\zeta|<1}\frac{\overline{\zeta}(\psi\circ\varphi(\zeta))}{2\zeta(\zeta-z)}d\xi d\eta
\end{align*}
with $\psi=-\rho(\frac{1}{2\mu+\lambda}\phi+\frac{i}{\mu}A)$,  $\gamma=\frac{1}{2\mu}(g_1-ig_2)+\frac{\rho}{2\mu}\left(\nu\phi-iA\right)\frac{\beta_\zeta}{|\beta_{\overline{\zeta}}|}$, where $\beta(\zeta,\overline{\zeta})=\varphi^{-1}(\zeta)\overline{\varphi^{-1}_\zeta(\zeta)}-1$ parametrizes the unit disk and $\nu=\frac{\lambda}{\lambda+2\mu}$ is the Poisson ratio.
\end{theorem}

\begin{proof} We start by considering the following problem in the disk $\DD$
\begin{align*}
    \pdzb v &= \frac{1}{2}\big(\psi\circ\varphi\big)\overline{\partial_z\varphi}\quad\text{ in }\DD,\\
    \partial_\nu v &= |\varphi_z|(\gamma\circ\varphi)\qquad\text{ on }\partial\DD.
\end{align*}
To obtain a solution we need to verify that the conditions for solvability stated in Theorem~\ref{thm:Begehr} are fulfilled. This is shown by a straightforward calculation
    \begin{align}
        I_1&=\frac{1}{2\pi i}\int_{|\zeta|=1}(\gamma\circ\varphi(\zeta))|\varphi_\zeta(\zeta)|\frac{d\zeta}{(1-\overline{z}\zeta)\zeta}+\frac{1}{4\pi i}\int_{|\zeta|=1}(\psi\circ\varphi(\zeta))\overline{\varphi_\zeta}\frac{d\overline{\zeta}}{1-\overline{z}\zeta}\\
        &=\frac{1}{2\pi i}\int_{\vert\zeta\vert=1}(\gamma\circ\varphi(\zeta))\left|\varphi_\zeta(\zeta)\right|\frac{d\zeta}{(1-\zb\zeta)\zeta}-\frac{1}{4\pi i}\int_{|\zeta|=1}\overline{\zeta}\big(\psi\circ\varphi(\zeta)\big)\overline{\varphi_\zeta}\frac{d\zeta}{\zeta(1-\zb\zeta)}\nonumber\\
        &=\frac{1}{2\pi i}\int_{|\zeta|=1}\left((\gamma\circ\varphi(\zeta))|\varphi_\zeta(\zeta)|-\frac{\overline{\zeta}}{2}(\psi\circ\varphi(\zeta))\overline{\varphi_\zeta}\right)\frac{d\zeta}{\zeta(1-\zb\zeta)}\nonumber\\
       & =\frac{1}{2\pi i}\int_{|\zeta|=1}\left(\zeta v_\zeta+\overline{\zeta} v_{\overline{\zeta}}-\frac{\overline{\zeta}}{2}(\psi\circ\varphi(\zeta))\overline{\varphi_\zeta}\right)\frac{d\zeta}{\zeta(1-\zb\zeta)}\nonumber \\ & =\frac{1}{2\pi i}\int_{|\zeta|=1}\frac{v_\zeta}{1-\zb\zeta}d\zeta ,
       \label{eq:conditionproof1}
    \end{align}
and analogously
    \begin{align}
        I_2 &= \frac{1}{2\pi}\int_{|\zeta|<1}\frac{\zb(\psi\circ\varphi(\zeta))\overline{\varphi_\zeta}}{(1-\zb\zeta)^2}d\xi d\eta\nonumber \\ & =\frac{1}{\pi}\int_{|\zeta|<1}\frac{\zb v_{\overline{\zeta}}}{(1-\zb\zeta)^2}d\xi d\eta
        =\frac{1}{2\pi i}\int_{|\zeta|=1}\frac{\zb v}{(1-\zb\zeta)^2}d\zeta, \label{eq:conditionproof2}
    \end{align}
where the last equality follows from the complex version of the Green theorem \cite{Begehr}, which is applicable from the regularity assumptions in the statement of this theorem. Set $I_3=I_1+I_2$. Then by \eqref{eq:conditionproof1} and \eqref{eq:conditionproof2} we get
    \begin{align*}
        I_3&=\frac{1}{2\pi i}\int_{|\zeta|=1}\bigg(\frac{v_\zeta}{1-\zb\zeta}+\frac{\zb v}{(1-\zb\zeta)^2}\bigg)d\zeta=\frac{1}{2\pi i}\int_{|\zeta|=1}\partial_\zeta\bigg(\frac{v}{1-\zb\zeta}\bigg)d\zeta=0.
    \end{align*}
It then follows from Lemma~\ref{lemma:transformedsimplyconnected} that the function $u=v\circ\varphi^{-1}(z)$ satisfies
\begin{align*}
    \partial_{\bar{w}}u &= \frac{\psi}{2}\text{ in }\Omega,\\
    \partial_\nu u &= \gamma\text{ on }\partial\Omega.
\end{align*}
 To finish the proof we shall write $\gamma$ in terms of the explicit components of the normal $\nn$. Introducing the function $\beta(\zeta,\overline{\zeta})=\varphi^{-1}(\zeta)\overline{\varphi^{-1}(\zeta)}-1$ we recall that the components of the unit normal are given by $\hat{n}_1=\operatorname{Re}\left(\frac{\beta_{\overline{\zeta}}}{|\beta_{\overline{\zeta}}|}\right)$ and $\hat{n}_2=\operatorname{Im}\left(\frac{\beta_{\overline{\zeta}}}{|\beta_{\overline{\zeta}}|}\right)$, resulting in
$$\gamma=\frac{1}{2\mu}(g_1-ig_2)+\frac{\rho}{2\mu}(\nu\phi-iA)\frac{\beta_\zeta}{|\beta_{\overline{\zeta}}|},$$
where we used that $\beta_\zeta=\overline{\beta_{\overline{\zeta}}}$. The theorem then follows analogously to Theorem~\ref{theorem:EquivalenceSimplyConnected} by showing that the physical field $\uu=[v_1,-v_2]^T$ satisfies Navier's equation with the correct boundary data.

\end{proof}

\section{The case of annular domains}
Let us now consider the case where our physical domain $\Omega$ is conformally equivalent to the annulus $\A=\{z\in\CC: r<|z|<1 \}$, rather than the unit disk as discussed in the previous section. This transition to an annular domain implies that the domain is no longer simply connected, and consequently, the Riemann mapping theorem does not apply. Although there is no direct analogue of the Riemann mapping theorem for annular domains, a result due to Vaitekhovich \cite{Vaitekhovich} ensures the existence of solutions to the Neumann problem for the inhomogeneous Cauchy–Riemann equations on $\A$. This corresponds precisely to Theorem~\ref{thm:Begehr} from the previous section.

Before we state the following theorem, it is worth recalling that the notation $W^{\alpha+1,\alpha}_{\zb}$ represents the set of Hölder continuous functions defined on the annulus that possess a continuous weak first-order $\zb$-derivative.

\begin{theorem}\label{theorem:annulussolution}
Let $\psi\in\mathcal{C}^{\alpha+1}(\A;\CC), 0<\alpha<1$, $\gamma\in\mathcal{C}(\partial\A;\CC)$, and assume that $\kappa=1$ on $|z|=1$, $\kappa=-1$ on $|z|=r$ and  $c\in\CC$ and $z_0\in \A$. The Neumann problem for the inhomogeneous Cauchy-Riemann equations in the annulus $\A$
\begin{align*}
    &\pdzb w=\psi\text{ in }\A\\
    &\kappa|z|\partial_{\nu_z}w=\gamma\text{ on }\partial\A\\
    & w(z_0)=c
\end{align*}
 is solvable if and only if for $z\in \A$
\begin{align*}
   & \frac{1}{2\pi i}\int_{\partial \A}\gamma(\zeta)\frac{d\zeta}{1-\zb\zeta}-\frac{1}{2\pi i}\int_{\partial \A}\psi(\zeta)\frac{\overline{\zeta}d\zeta}{1-\zb\zeta}+\frac{1}{\pi}\int_{\A}\psi(\zeta)\frac{d\xi d\eta}{(1-\zb\zeta)^2}=0\\
   & \frac{1}{2\pi i}\int_{\partial \A}\gamma(\zeta)\frac{d\zeta}{r^2-\zb\zeta}-\frac{1}{2\pi i}\int_{\partial \A}\psi(\zeta)\frac{\overline{\zeta}d\zeta}{r^2-\zb\zeta}+\frac{r^2}{\pi}\int_{\A}\psi(\zeta)\frac{d\xi d\eta}{(r^2-\zb\zeta)^2}=0.
\end{align*}
Moreover, if $\gamma$ and $\psi$ satisfy the condition:
\begin{equation*}
    \frac{1}{2\pi i}\int_{\partial \A}\left(\gamma(\zeta)-\overline{\zeta}\psi(\zeta)\right)\frac{d\zeta}{\zeta}=0 ,
\end{equation*}
then the solution is unique, single-valued function given by
\begin{align*}
    w(z)  =&c-\frac{1}{2\pi i}\int_{|\zeta|=1}\left(\gamma(\zeta)-\overline{\zeta} \psi(\zeta)\right)\log\left(\frac{1-z\overline{\zeta}}{1-z_0\overline{\zeta}}\right)\frac{d\zeta}{\zeta}\\
    & +\frac{1}{2\pi i}\int_{|\zeta|=r}\left(\gamma(\zeta)-\overline{\zeta}\psi(\zeta)\right)\log\left(\frac{z\overline{\zeta}-r^2}{z_0\overline{\zeta}-r^2}\right)\frac{d\zeta}{\zeta}\\
    & -\frac{1}{\pi}\int_{\A} \psi(\zeta)\frac{z-z_0}{(\zeta-z_0)(\zeta-z)}d\xi d\eta.
\end{align*}
\end{theorem}
\begin{proof}
See~\cite[Theorem 4.4]{Vaitekhovich}.
\end{proof}

\begin{theorem}\label{theorem:EquivalenceAnnulus}
    Let $\A$ denote the unit annulus with inner radius $r$ and $\mu,\lambda$ its associated Lamé parameters. Furthermore, let $\phi$, $\pmb{A}=A\zz$ be potentials, where $\phi,A$, are $\mathcal{C}^2$ in a neighborhood of $\A$ and $\pmb{F}=~\nabla\phi+\nabla\times\pmb{A}$ denote the associated specific loading force. Finally, let the boundary traction $\GG=[g_1, g_2]^T$ be $\mathcal{C}^1$ in a neighborhood of $\partial\A$. Then, the existence of the displacement field $\uu = [u_1, u_2]^T$ that solves
    \begin{align*}
        -&\nabla\cdot\sigma(\uu)=\rho\pmb{F}\text{ in }\A\\
        &\sigma(\uu)\cdot\nn=\GG\text{ on }\partial \A\\
        & \uu(z_0)=\pmb{0}\, ,
    \end{align*}
    where $p_0\in \A$, is equivalent to that the function $v=u_1-iu_2$ solves
    \begin{align*}
        \pdzb v&=\frac{\psi}{2}\text{ in }\A\\
        \kappa|z|\partial_{\nu} v&=\gamma\text{ on }\partial\A\\
        v(z_0)&=0,
    \end{align*}
    with $\psi=-\frac{\phi}{2\mu+\lambda}-i\frac{A}{\mu}$, $\gamma=\frac{\kappa}{2\mu|\zeta|}(g_1-ig_2)+\frac{\kappa\rho}{2\mu}\frac{\overline{\zeta}}{|\zeta|}\big(\nu\phi-iA\big)$, where $\kappa$ is a constant defined as $\kappa=1$ for $|\zeta|=1$, $\kappa=-1$ for $|\zeta|=r$
    and $\nu=\frac{\lambda}{\lambda+2\mu}$ is the Poisson ratio.
\end{theorem}
\begin{proof}
The proof follows the arguments of Theorem~\ref{theorem:EquivalenceSimplyConnected}, but we need to use Theorem~\ref{theorem:annulussolution} instead of Theorem~\ref{thm:Begehr}. It should be noted that for these arguments to go through one has to partition the annulus along the horizontal axis and show the arguments separately.
\end{proof}

Similarly, as in the proof of
Theorem~\ref{theorem:SimplyConnectedMain},  one can show that we can extend Theorem~\ref{theorem:EquivalenceAnnulus}. Therefore, we can establish the existence of a unique solution to the Neumann problem for the plane elasticity equations~\eqref{eq:elasticityproblem}  in domains that are conformally equivalent to $\A$. This insight leads to the main result of this section, Theorem~\ref{theorem: SolutionAnnularDomain}.

\begin{theorem}
\label{theorem: SolutionAnnularDomain}
Let $\varphi$ be a biholomorphic mapping from a neighborhood of $\overline{\A}$ into a neighborhood of $\overline{\Omega}$, such that $\varphi(\A)=\Omega$, $\varphi(\partial \A)=\partial\Omega$ and  $\varphi(z_0)=p_0$ where $\partial\Omega$ is a Jordan curve, $z_0\in\A$ and $p_0\in\Omega$. Furthermore, $\Omega$ has associated Lamé parameters $\mu,\lambda$. Furthermore, let $\phi$, $\pmb{A}=A\zz$ be potentials, where $\phi,A$, are $\mathcal{C}^2$ in a neighborhood of $\Omega$ and $\pmb{F}=~\nabla\phi+\nabla\times\pmb{A}$ denote the associated specific loading force. Finally, let the boundary traction  $\GG=[g_1, g_2]^T$  be $\mathcal{C}^1$ in a neighborhood of $\partial\Omega$.  Then, a unique solution exists to the plane elasticity problem
\begin{align*}
    -\nabla\cdot\sigma(\uu)&=\rho\pmb{F}\text{ in }\Omega\\
\sigma(\uu)\cdot\nn&=\GG\text{ on }\partial\Omega\\
    \uu(p_0)&=\pmb{0}.
\end{align*}
Moreover, the solution is given by $\uu=[u_1,u_2]^T$, where $u_1,u_2$ are the positive real and negative imaginary parts respectively of $v\circ\varphi^{-1}$ where $v$ is given by

\begin{align*}
    v(z)  =&\frac{1}{2\pi i}\int_{|\zeta|=1}\left(|\varphi_\zeta|(\gamma\circ\varphi(\zeta))-\frac{\overline{\zeta}}{2}(\psi\circ\varphi(\zeta))\overline{\varphi_\zeta}\right)\log\left(\frac{1-z\overline{\zeta}}{1-z_0\overline{\zeta}}\right)\frac{d\zeta}{\zeta}\\
    & +\frac{1}{2\pi i}\int_{|\zeta|=r}\left(|\varphi_\zeta|(\gamma\circ\varphi(\zeta))-\frac{\overline{\zeta}}{2}(\psi\circ\varphi(\zeta))\overline{\varphi_\zeta}\right)\log\left(\frac{z\overline{\zeta}-r^2}{z_0\overline{\zeta}-r^2}\right)\frac{d\zeta}{\zeta}\\
    & -\frac{1}{\pi}\int_{\A} \frac{(\psi\circ\varphi(\zeta))\overline{\varphi_\zeta}(\zeta)}{2}\frac{z-z_0}{(\zeta-z_0)(\zeta-z)}d\xi d\eta.
\end{align*}
Here, $\psi=-\frac{\phi}{2\mu+\lambda}-i\frac{A}{\mu}$, $\gamma=\frac{\kappa}{2\mu|\zeta|}(g_1-ig_2)+\frac{\kappa\rho}{2\mu|\zeta|}\frac{\beta_\zeta}{|\beta_{\overline{\zeta}}|}\big(\nu\phi-iA\big)$, $\beta(\zeta,\overline{\zeta})=\varphi^{-1}(\zeta)\overline{\varphi^{-1}(\zeta)}-1$, $\kappa=1$ for $|\zeta|=1$, $\kappa=-1$ for $|\zeta|=r$ and  $\nu=\frac{\lambda}{\lambda+2\mu}$ is the Poisson ratio.
\end{theorem}

\section{Applying the new method to classic cases}\label{sec:Appl}

In this section, we demonstrate the efficacy of Theorem~\ref{theorem:SimplyConnectedMain} and~\ref{theorem: SolutionAnnularDomain} by exploring its application to specific traditional cases in plane elasticity. These cases have been carefully selected, representing scenarios where geometric complexities and applied forces present distinct challenges. By providing explicit solutions for the displacement field in cases such as a notch problem in a cardioid domain, a ring domain with a shifted hole, and a gear-like structure, we aim to highlight the advantages and versatility of our method in contrast to traditional approaches.

In all examples, we maintain the same material, Titanium--$\beta$--C, due to its wide industrial use in sectors such as aerospace, automotive, and medical technology, given its high strength-to-weight ratio, excellent corrosion resistance, and biocompatibility. This alloy is characterized by Lamé parameters $\mu=34375$ $\mathrm{MPa}$ and $\lambda=43750$ $\mathrm{MPa}$, and density $\rho = 4820$ $\mathrm{kg\,  m}^{-3}$. We employ a scaling constant $c$ to modulate the applied body force and select the functions $\gamma$ and $\psi$ to enhance the comprehensibility of the solutions. To simplify the presentation of the examples, we assume that the boundary traction $\GG$ is given implicitly rather than explicitly specifying it.

We start in a cardioid domain, often encountered in contexts where structural weakening and maximum load before fracturing are of interest. A practical example includes the construction of micromechanical components in silica, which contain near-atomic notches inherent in their crystalline structure \cite{Busch}. However, the sharp corner introduces problematic behavior of stress fields in the Airy or biharmonic framework, a recognized issue for over half a century \cite{Williams}. While specialized plate-theories and methods have found some success with specific boundary conditions and loads, a common approach is to consider a simplified or non-existent traction force in a biharmonic equation, as solved by Williams in the referenced paper, leading to the so-called Williams expansion \cite{Helsing}. However, this approach can be cumbersome and tends to diverge from the physical problem. Here, we demonstrate an example using our method, leveraging only the force potentials and boundary traction to directly evaluate the displacement field.

Another example covers domains of rings with eccentric or shifted holes. This has been a longstanding topic of interest due to its applicability in engineering. Most authors consider the application of constant forces, either applied pointwise to the outer or inner circle \cite{Gupta1,Gupta2} or diametrically and normally applied pressure \cite{Alaci2,Desai,Solovev}. Unlike these studies, our method applies to continuous forces, not pointwise ones. Here, we illustrate its usage on continuous force distributions in an eccentric ring. The outlined method's simplicity compared to common bipolar coordinates application \cite{Khomasuridze} is another advantage, as it directly utilizes the problem's physical data and an integral representation instead of the power-series method.

\clearpage

\begin{example}\label{Ex:cardioid}
In this example, we consider the physical domain $\Omega$ as a cardioid, a domain with a single inward corner where we expect to observe a stress concentration. Given the anticipated singularity in the solution at this corner, we initially take the domain $\Omega_a$ as the image of $\DD$ under the transformation $\varphi_a(\zeta)=\frac a2\zeta^2+\zeta$ with $0<a<1$. The canonical solution is then given by $v(\zeta)=\lim_{a\rightarrow 1^-}v_a(\zeta)$.

The force is characterized by the potentials:
\begin{align*}
    \phi&=-\frac{(2\mu+\lambda)c}{\rho}(x^3-3xy^2),\\
    \pmb{A}&=-\frac{\mu c}{\rho}(3x^2y-y^3),
\end{align*}
leading to $\psi=c\zeta^3$. The complex traction potential is assumed as:
\[
\gamma_a(w)=\frac {\varphi_a^{-1}(w)}{|\varphi'_a(\varphi_a^{-1}(w))|}=\frac 1a\, \frac {-1+\sqrt{1+2aw}}{\sqrt{|1+2aw|}},
\]
where $\sqrt{\cdot}$ denotes the complex branch of square root satisfying $\sqrt{1}=1$.The solution in the domain $\Omega_a$ is then given by $u_a(\zeta)=v_a(\varphi_a^{-1}(\zeta))$, where
\begin{multline*}
v_a(z)=\frac{c}{2}\bigg( \frac {a}{2}z^2+z\bigg)^3\bigg(\frac {a}{2}\bar z^2+\bar z\bigg)+cz\\
-\frac {c}{16}\bigg(\frac 75a^3z^5+\big(a^4+9a^2\big)z^4+(7a^3+20a)z^3+(18a^2+16)z^2+20az\bigg).
\end{multline*}

As $a \to 1^-$, the physical domain $\Omega$ takes on the shape of a cardioid. In this case, a normal vector cannot be defined at the cardioid's cusp, creating a problem for our plane elasticity problem as laid out in Theorem~\ref{theorem:SimplyConnectedMain}. While our theorem holds true for all other points in a neighborhood of $\Omega$, it breaks down at the cardioid's cusp.

Nevertheless, if one insists on mathematically extending the model to the cardioid's cusp, the optimal method would be to interpret it as follows: In $\Omega_1$, for the data $\gamma_1(w)=\frac {-1+\sqrt{1+2w}}{\sqrt{|1+2w|}}$, the solution is given by by $u(\zeta)=v(-1+\sqrt{1+2\zeta})$, where
\begin{multline*}
v(z)=\lim_{a\rightarrow 1^-}v_a(z)=\frac{c}{2}\left( \frac {1}{2}z^2+z\right)^3\left(\frac {1}{2}\bar z^2+\bar z\right)\\ +cz-
\frac {c}{16}\left(\frac 75z^5+10z^4+27z^3+34z^2+20z\right).
\end{multline*}

From an applied perspective, this problem is less significant, since one can simply choose an $a$ value sufficiently close to 1. In Fig.~\ref{App1}, we selected $a=0.99$.

\begin{figure}[H]
     \centering
     \begin{subfigure}[b]{0.45\textwidth}
         \centering
         \includegraphics[width=\textwidth]{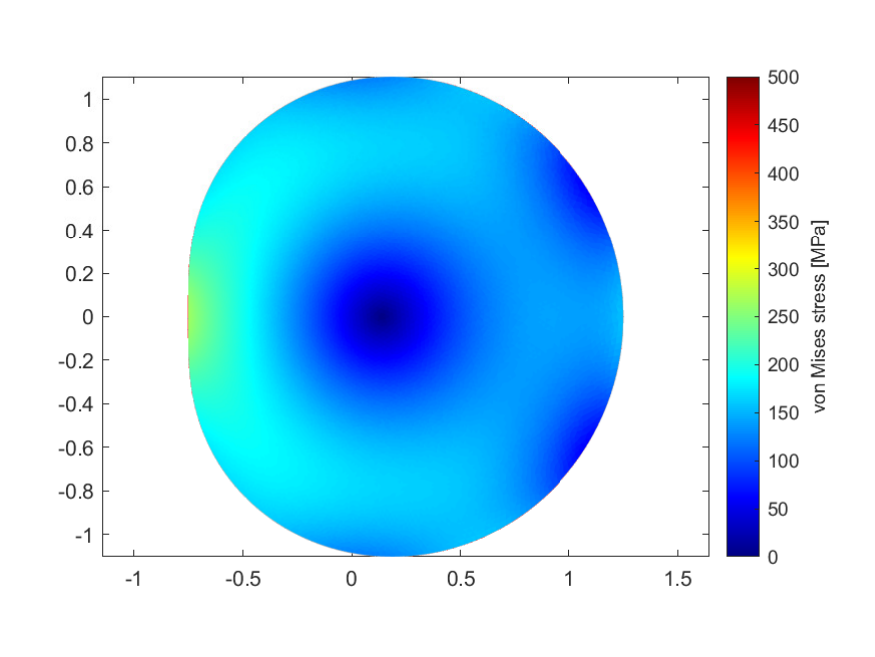}
         \caption{$a=0.5$ and $c=10^{-3}$}
         \label{App1a}
     \end{subfigure}
     \hfill
     \begin{subfigure}[b]{0.45\textwidth}
         \centering
         \includegraphics[width=\textwidth]{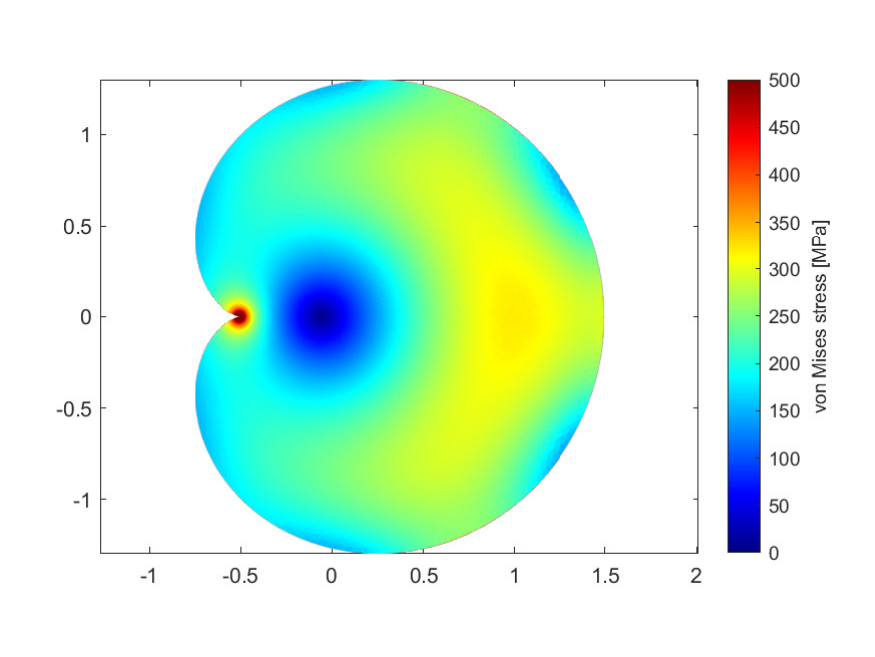}
         \caption{$a=0.99$ and $c=10^{-3}$}
         \label{App1b}
     \end{subfigure}
     \caption{Internal distribution of von Mises stress for different values of $a$}
     \label{App1}
\end{figure}
 \hfill{$\Box$}
\end{example}

\begin{example}
\label{ex: eccentric ring}
In this example, we shall consider an eccentric annulus. Let $\Omega_b$ be the image of the annulus $\mathbb{A}_{\frac{1}{2}}$ under the mapping $\varphi_b=\frac{z-b}{1-bz}$ where $0<b<1$. The new domain $\Omega_b$ is then defined as $\Omega_b=B(0,1)\setminus \overline{B(z_b,r_b)}$, with center $z_b=\frac {3b}{b^2-4}$ and radius $r_b=\frac{2(1-b^2)}{4-b^2}$. Assume the potentials of the force acting on the domain are given by:
\begin{align*}
    \phi(z)&=-\frac{(2\mu+\lambda)c}{\rho}(x^2-y^2)\\
    \pmb{A}&=-\frac{2\mu c}{\rho}xy,
\end{align*}

which implies that $\psi=c\zeta^2$. Let us also consider the case when the complex traction potential is given by
$$\gamma(w)=\frac{c(1-b^2)(w+b)}{|1+bw|^2(1+bw)}+\frac{c|w|^2(w+b)}{1+b\overline{w}}+\frac{cw^2(\overline{w}+b)}{2(1+bw)}.
$$

It can then readily be shown that the conditions for solvability are satisfied and we obtain a canonical solution as $u(\zeta)=v(\varphi_b^{-1}(\zeta))$, where
$$
v(z)=\frac{c}{2}\frac{(z-b)|z-b|^2}{(1-bz)|1-bz|^2}+c(z-b),
$$
which yields $v(b)=0$. Setting $b=0.5$ and $c=10^{-3}$ gives us the internal distribution of von Mises stress as shown in Fig.~\ref{App2}.

\clearpage

\begin{figure}[H]
    \centering
\includegraphics[width=0.75\textwidth]{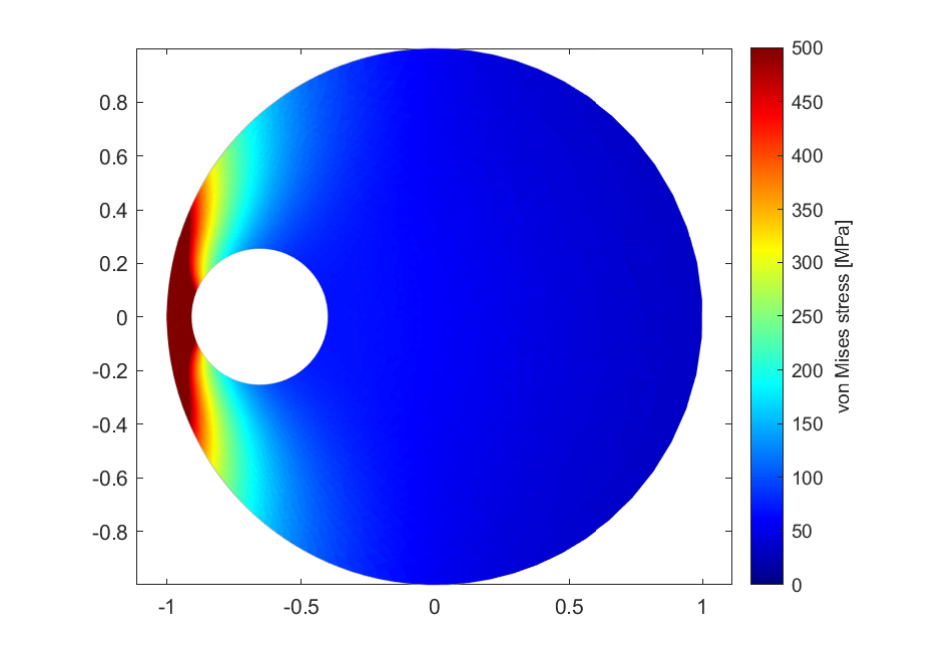}
    \caption{Internal distribution of von Mises stress due to a linear load with scaling constant $c=10^{-3}$.}
    \label{App2}
\end{figure}
 \hfill{$\Box$}
\end{example}

\begin{example}
Consider the gear-like domain $\Omega_a$, defined as the image of an annulus $\A$ with inner radius $r=1/2$ under the map $\varphi_a=\frac{a\zeta^4}{4}+\zeta$ with $0<a<1$. This domain represents an epitrochoid. We assume that the complex boundary traction potential, $\gamma_a(w)$, is given by the formula:
\begin{multline*}
\gamma_a(w)=\\ \frac {2c\varphi_a^{-1}(w)+2c|w|^2\varphi_a^{-1}(w)(a(\varphi_a^{-1}(w))^3+1)+cw^2\overline{\varphi_a^{-1}(w)}(a\overline{(\varphi_a^{-1}(w))}^3+1)}{2|a(\varphi_a^{-1}(w))^3+1|},
\end{multline*}
where $\varphi_a^{-1}(w)$ is defined as:
\[
\varphi_a^{-1}(w)=\frac 12\left(-\sqrt{Y(w)}+\sqrt{-2Y(w)+\frac {8}{a\sqrt{2Y(w)}}}\right),
\]
and $Y(w)$ and $X(w)$ are given by:
\[
Y(w)=X(w)-\frac {4w}{3aX(w)}, \ \ X(w)=a^{-\frac 23}\sqrt[3]{1+\sqrt{1+\frac {64aw^3}{27}}}.
\]
We consider the action of a constant force, with potentials $\phi(z)$ and $\pmb{A}$ defined as:
\begin{align*}
    \phi(z)&=-\frac{(2\mu+\lambda)c}{\rho}(x^2-y^2),\\
    \pmb{A}&=-\frac{2\mu c}{\rho}xy,
\end{align*}
yielding $\psi=c\zeta^2$. The canonical solution is then given by $u_a(\zeta)=v_a(\varphi_a^{-1}(\zeta))$, where
$$
v_a(z)=\frac c2\left|\frac {a}{4}z^4+z\right|^2\left(\frac {a}{4}z^4+z\right)+cz.
$$

Here, we face a situation similar to Example~\ref{Ex:cardioid}. The same reasoning applies, but we forego the details to avoid repetition and suggest revisiting Example~\ref{Ex:cardioid} for understanding.

The resulting internal distribution of von Mises stress with $c=10^{-3}$ and the choices $a=0.5$ and $a=0.99$ are shown in Fig.~\ref{App3}.
\begin{figure}[H]
     \centering
     \begin{subfigure}[b]{0.45\textwidth}
         \centering
         \includegraphics[width=\textwidth]{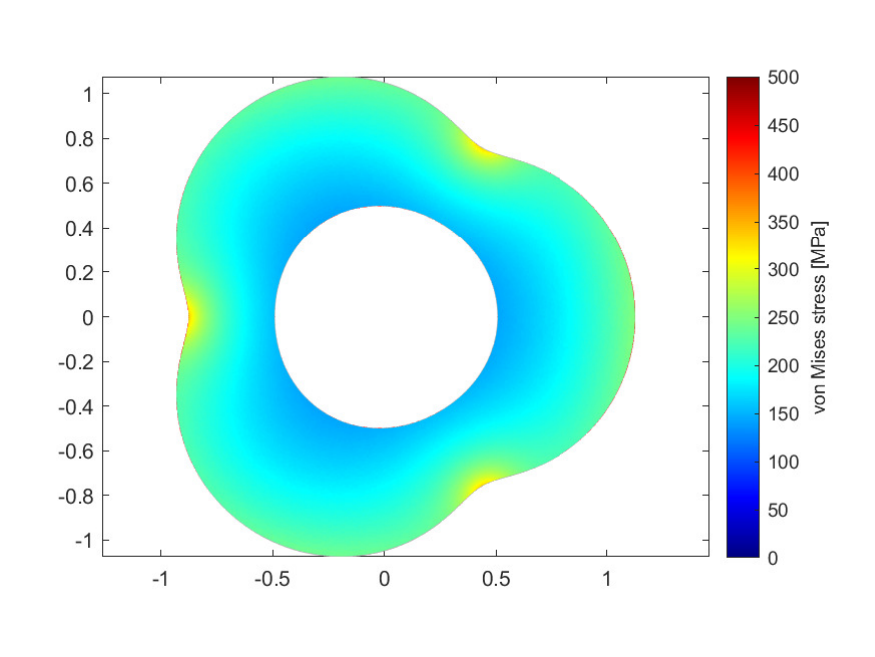}
         \caption{$a=0.5$ and $c=10^{-3}$}
         \label{App3a}
     \end{subfigure}
     \hfill
     \begin{subfigure}[b]{0.45\textwidth}
         \centering
         \includegraphics[width=\textwidth]{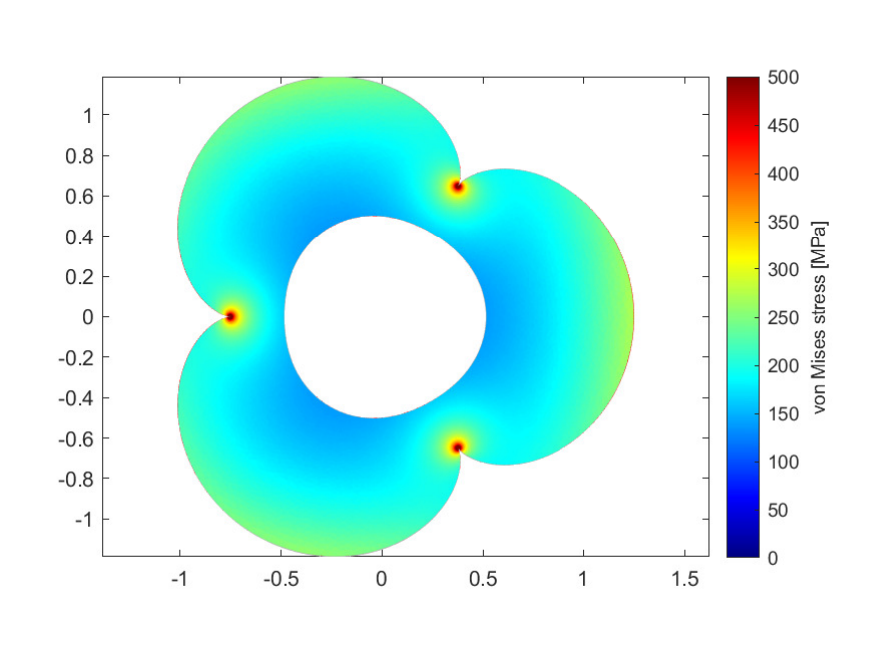}
         \caption{$a=0.99$ and $c=10^{-3}$}
         \label{App3b}
     \end{subfigure}
     \caption{Internal distribution of von Mises stress for different values of $a$}
     \label{App3}
\end{figure}
\hfill{$\Box$}
\end{example}

\section{Conclusion}

In this paper, we presented a methodology that seeks to address challenges in isotropic planar elastostatics based on the Neumann problem for the inhomogeneous Cauchy-Riemann equations.

\medskip

\noindent Features of our method

\begin{itemize}
    \item \emph{Direct Correspondence}: Our method establishes the connection between the plane elasticity problem described in~\eqref{eq:elasticityproblem} and the Neumann problem for the inhomogeneous Cauchy-Riemann equations.

    \item \emph{Flexibility and Scope}: While many techniques are designed for specific force classes or domains, our method aims to provide a broader perspective by setting conditions on a wider range of domains and data smoothness criteria.

    \item \emph{Focus on the Displacement Field}: Instead of prioritizing the stress tensor, which is common in many established methods, our approach attempts to provide an explicit formula for the displacement field.

    \item \emph{Simpler Approach}: By avoiding the need for a rational conformal map, our methodology provides a more general approach. This is further enhanced by minimizing intermediate constructs, which are prevalent in methods such as the real integral equation method and the Somigliana method.

    \item \emph{Range of Applicability}: Our method has been tested on a variety of domains, including cardioid structures and gear-like configurations, showcasing its potential applications.

\end{itemize}

In conclusion, our research introduces a novel approach that builds on existing concepts in isotropic planar elastostatics and offers potential for broader applications.

\end{document}